\documentclass[sigconf, nonacm]{acmart}

\usepackage{enumitem}
\setlist[itemize]{leftmargin=18pt}

\usepackage{thmtools, thm-restate}
\declaretheorem{theorem}

\newtheorem{problem}{Problem}
\newcommand{\secref}[1]{\S\ref{#1}}

\usepackage{mathtools}

\newcommand{\cost}{\operatorname{cost}}

\newcommand{\STS}{\textsc{STS}}

\usepackage{tikz}
\usetikzlibrary{
  arrows.meta,
  positioning,
  calc,
  decorations.pathreplacing,
  shapes.symbols
}

\newcommand{\sparagraph}[1]{\vspace{1mm}\noindent {\bf #1}}

\newcommand\vldbdoi{XX.XX/XXX.XX}
\newcommand\vldbpages{XXX-XXX}

\newcommand\vldbvolume{14}
\newcommand\vldbissue{1}
\newcommand\vldbyear{2020}

\newcommand\vldbauthors{\authors}
\newcommand\vldbtitle{\shorttitle} 

\newcommand\vldbavailabilityurl{https://github.com/Hedi-Chehaidar/optfsst}

\newcommand\vldbpagestyle{plain}

\usepackage[switch]{lineno}

\newcommand{\optfsst}{\textsc{OptFSST}}
\newcommand{\optfsstxx}{\textsc{OptFSST12}}
\newcommand{\fsst}{\textsc{FSST}}
\newcommand{\fsstxx}{\textsc{FSST12}}

 \usepackage{multirow} 
\usepackage{listings}

\usepackage{subcaption}
\usepackage{tcolorbox}
\usepackage{xcolor}
\usepackage[dvipsnames]{xcolor}

\newcommand{\hlstrut}{\rule[-0.15ex]{0pt}{1.4ex}}

\newcommand{\gray}[1]{\colorbox{gray!15}{\hlstrut #1}}

\newcommand{\orange}[1]{\colorbox{BurntOrange!18}{\hlstrut #1}}
\newcommand{\green}[1]{\colorbox{ForestGreen!15}{\hlstrut #1}}
\newcommand{\purple}[1]{\colorbox{Purple!15}{\hlstrut #1}}
\newcommand{\red}[1]{\colorbox{BrickRed!15}{\hlstrut #1}}
\newcommand{\pblue}[1]{\colorbox{ProcessBlue!15}{\hlstrut #1}}
\newcommand{\pink}[1]{\colorbox{Magenta!18}{\hlstrut #1}}

\definecolor{codekeyword}{RGB}{0,0,180}
\definecolor{codecomment}{RGB}{90,90,90}
\definecolor{codestring}{RGB}{160,60,60}
\definecolor{codenumber}{RGB}{120,120,120}
\definecolor{codebackground}{RGB}{248,248,248}
\begin{document}
\title{OptFSST: Optimized FSST String Compression}

\settopmatter{authorsperrow=4}

\author{Hedi Chehaidar}
\affiliation{%
  \institution{Technical University of Munich}
  \country{Germany}
}
\email{hedi.chehaidar@tum.de}

\author{Mihail Stoian}
\orcid{}
\affiliation{%
  \institution{University of Technology Nuremberg}
  \country{Germany}
}
\email{mihail.stoian@utn.de}

\author{Moritz Stargalla}
\orcid{}
\affiliation{%
  \institution{University of Technology Nuremberg}
  \country{Germany}
}
\email{moritz.stargalla@utn.de}

\author{Andreas Kipf}
\orcid{}
\affiliation{%
  \institution{University of Technology Nuremberg}
  \country{Germany}
}
\email{andreas.kipf@utn.de}

\begin{abstract}
Strings account for a substantial fraction of data in modern analytical systems, making lightweight compression with fast random access an important building block for efficient query processing. Fast Static Symbol Table (\fsst{}) addresses this need by replacing frequent byte sequences with compact codes while preserving independent decompression of individual strings. However, \fsst{}'s compression effectiveness is limited by its greedy symbol selection and greedy encoding strategy, leaving encoding gains on the table.

We present \optfsst{}, an optimized \fsst{} variant that improves its compression factors while preserving its static-symbol-table design and random-access decompression. \optfsst{} optimally encodes the text using dynamic programming given a symbol table. Additionally, we show that a generalized version of the symbol-table selection problem is NP-hard when the alphabet is part of the input, motivating heuristic table construction for field-level compressors. Hence, we add in \optfsst{} (i) an additional frequency counter that accelerates the discovery of longer symbols and (ii) a pruning strategy that removes redundant and conflicting symbol candidates during table construction. We also extend the same techniques to \fsstxx{}, yielding \optfsstxx{}.

Our evaluation on 92 real-world string datasets shows that \optfsst{} improves the compression factors of \fsst{} and \fsstxx{} by up to $47.7\%$ and $91.5\%$, with an average improvement of $7.3\%$ and $17.0\%$, respectively, while retaining the fine-grained random-access properties. Notably, \optfsstxx{} improves \fsstxx{}'s decompression speed by $1.2\times$ on average.
\end{abstract}

\maketitle

\pagestyle{\vldbpagestyle}
\begingroup\small\noindent\raggedright\textbf{VLDB Workshop Reference Format:}\\
VLDB 2026 Workshop: 17th International Workshop on Accelerating Analytics and
Data Management Systems Using Modern Processor and Storage
Architectures (ADMS26).
\endgroup
\begingroup
\renewcommand\thefootnote{}\footnote{\noindent
This work is licensed under the Creative Commons BY-NC-ND 4.0 International License. Visit \url{https://creativecommons.org/licenses/by-nc-nd/4.0/} to view a copy of this license. For any use beyond those covered by this license, obtain permission by emailing \href{mailto:info@vldb.org}{info@vldb.org}. Copyright is held by the owner/author(s). Publication rights licensed to the VLDB Endowment. \\
\raggedright Proceedings of the VLDB Endowment. ISSN 2150-8097. \\
}\addtocounter{footnote}{-1}\endgroup

\ifdefempty{\vldbavailabilityurl}{}{
\vspace{.3cm}
\begingroup\small\noindent\raggedright\textbf{PVLDB Artifact Availability:}\\
The source code, data, and/or other artifacts have been made available at \url{\vldbavailabilityurl}.
\endgroup
}

\section{Introduction}
\label{sec:introduction}
 
\begin{figure}[t]
    \centering
    \includegraphics[width=\linewidth]{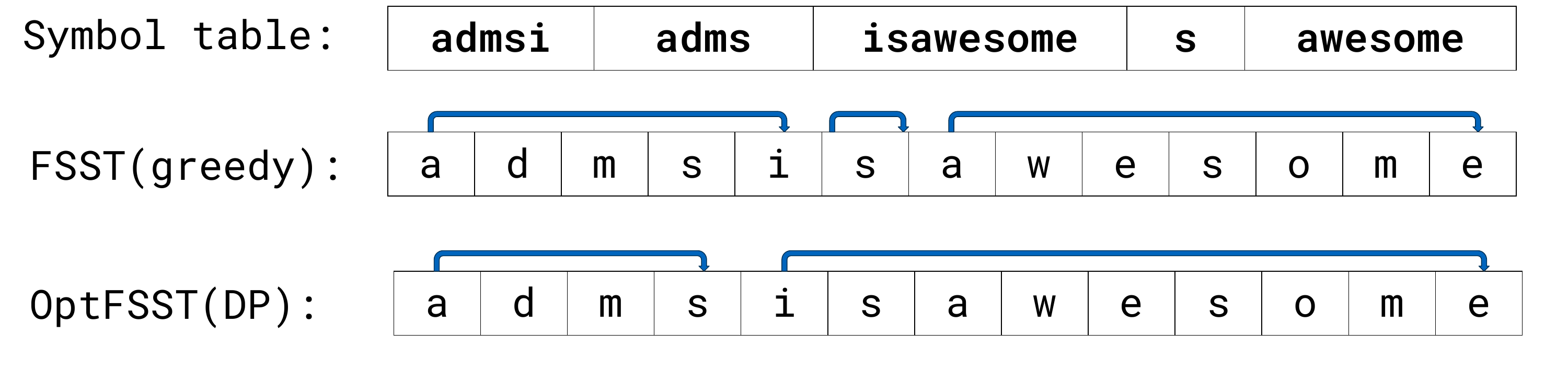}
    \caption{Example where \fsst{}'s greedy longest-match encoding is suboptimal. \optfsst{} yields the optimal compressed size given the symbol table.}
    \label{fig:greedy-dp}
\end{figure}

\sparagraph{Strings in Analytical Workloads.} Strings are a central data type in analytical data management systems. They occur as URLs, identifiers, file paths, log messages, categorical attributes, free-form text, and semi-structured values originating from web, cloud, and enterprise applications. Studies of production analytical workloads and classic column-store systems show that strings can represent a substantial part of stored data and can dominate memory consumption in compressed columnar layouts~\cite{vanrenen2024redshift,boncz1995monet,abadi2008columnstores}. Consequently, effective string compression is important not only for reducing storage size, but also for improving cache behavior, lowering memory-bandwidth pressure, and enabling efficient query processing on compressed data~\cite{zukowski2006superscalar,holloway2007barter,lang2016datablocks}.

This memory-side pressure makes compact in-memory string representations
valuable even when obtaining them requires additional compression-side work,
provided that decompression remains fast and fine-grained.

\sparagraph{Block-Based Compression.} General-purpose compressors such as Zstandard~\cite{facebook2015zstd} and LZ4~\cite{collet2011lz4} are highly effective on large contiguous byte streams, but do not support efficient access to individual strings. This motivates string compressors that preserve fine-grained random access while still reducing memory footprint.

\sparagraph{FSST.} Fast Static Symbol Table (\fsst{}) was designed for this setting~\cite{boncz2020fsst}. It replaces frequent byte sequences with compact codes from a static symbol table, allowing individual strings to be decompressed independently. This
combination of compact strings, fast decompression, and random access makes
FSST attractive for columnar storage formats and database systems. \fsst{}-inspired dictionary compression is used or explored in systems and formats such as BtrBlocks~\cite{kuschewski2023btrblocks}, FastLanes~\cite{afroozeh2025fastlanes}, F3~\cite{zeng2025f3}, Lance~\cite{pace2025lance}, Vortex~\cite{vortex2026}, and DuckDB's \texttt{DICT\_FSST}~\cite{bruineman2025dictfsst} compression.

\sparagraph{Limitations.} The compression effectiveness of \fsst{} depends on the quality of its symbol table and on the encoding decisions made with that table. The original \fsst{} construction relies mainly on heuristics: it repeatedly compresses a sample, counts symbol and pair frequencies, and selects candidates according to local gain estimates. The final encoder then greedily emits the longest matching symbol at the current input position. These choices are fast, but they can be suboptimal. A locally frequent short symbol may block a better sequence of symbols, overlapping candidates can overestimate their true gain, and redundant candidates can occupy entries in the limited symbol table.

\sparagraph{Example.} Figure~\ref{fig:greedy-dp} gives a small example of the limitation caused by a greedy encoding. \fsst{} selects the longest matching symbol at each position, even if this local decision prevents a better global segmentation of the remaining string. In the example, the symbol ``\texttt{admsi}'' is attractive locally, but choosing it leaves the remainder to be encoded using more symbols, resulting in 3 symbols in total. \optfsst{} finds the optimal strategy so as to minimize the compressed size, in this case 2 symbols.

\sparagraph{OptFSST \& OptFSST12.} We present \optfsst{}, an improved \fsst{} variant that targets these limitations while preserving \fsst{}'s byte-oriented format, static-symbol-table design, random-access property, and fast decompression path with a clear improvement in decoding speed in the case of \optfsstxx{}. \optfsst{} improves both encoding and symbol-table construction: it uses dynamic programming to choose globally better encodings for a fixed symbol table, and it refines table construction with longer-symbol frequency counting and conflict-aware pruning. We also apply the same ideas to \fsstxx{}, yielding \optfsstxx{}.

\sparagraph{Contributions.} In summary, we make the following contributions:
\begin{itemize}
  \item[1.] We introduce \optfsst{}, an optimized \fsst{} variant that computes an optimal encoding for a given symbol table using dynamic programming while preserving \fsst{}'s decoder path and random-access property.
  \item[2.] We show that selecting the optimal symbol table is NP-hard,\footnote{When the alphabet is part of the input; we refer the reader to~\secref{sec:np-hardness} for more details.} motivating heuristic table construction, and introduce longer-symbol frequency counting and conflict-aware pruning as a lightweight correction that improves many columns and reduces redundant symbol-table entries.
  \item[3.] We extend the same techniques to \fsstxx{} and evaluate \optfsst{} and \optfsstxx{} on 92 real-world string datasets, improving compression factors by up to $47.7\%$ and $91.5\%$, with average improvements of $7.3\%$ and $17.0\%$, respectively. \optfsstxx{} improves the average decompression speed by $1.2\times$ over \fsstxx{}.
\end{itemize}

\sparagraph{Paper Organization.} The remainder of this paper is organized as follows. In 
\secref{sec:related-work}, we discuss related work on string compression, block-based compression, dictionary compression, and \fsst{} extensions. Then, in \secref{sec:approach}, we present the \optfsst{} design and implementation which combines dynamic-programming-based encoding, longer-symbol frequency counting, and conflict-aware symbol pruning. Notably, to motivate the introduced heuristics, we show in \secref{sec:np-hardness} that finding the optimal symbol table is indeed NP-hard. We evaluate in \secref{sec:evaluation} \optfsst{} and \optfsstxx{} on real-world string datasets and quantify the trade-off between compression factor and compression speed, and conclude in \secref{sec:conclusion}.
\section{Related Work}
\label{sec:related-work}

\sparagraph{Compression in Analytical Systems.} Compression is a core component of modern analytical database systems. Besides reducing storage footprint, compressed representations improve cache locality, reduce memory bandwidth pressure, and can enable query execution directly on compressed data~\cite{zukowski2006superscalar,holloway2007barter,lang2016datablocks,raman2013blu}. Columnar systems in particular benefit from lightweight compression because values of the same type are stored contiguously and often exhibit repeated structure~\cite{abadi2008columnstores}. Much of the classic work on database compression focuses on numeric columns, dictionary encoding, run-length encoding, bit-packing, and other integer-oriented schemes. String data, however, remains more challenging: strings are variable-length, often have high cardinality, and may contain repeated substrings even when full values are unique.

\sparagraph{Block-Based Compression.} General-purpose compressors such as LZ4~\cite{collet2011lz4}, Zstandard~\cite{facebook2015zstd}, and Snappy~\cite{snappy} are widely used in storage and data processing systems because they provide strong throughput-compression trade-offs on large byte streams. All of the mentioned compressors are based on the Lempel-Ziv family of compression techniques~\cite{ziv1977universal}. LZ4 prioritizes speed and uses a simple block-oriented format with greedy match finding, while Zstandard combines LZ-style matching with entropy coding and a tunable compression-speed trade-off. Snappy, originally released by Google in 2011, emphasizes very fast compression and decompression with moderate compression ratios, making it a common choice for latency-sensitive storage and data-processing workloads.

These compressors are highly effective for sequential access and large blocks, but their block-oriented nature conflicts with fine-grained random access. Accessing one string typically requires decompressing the block that contains it, which can be wasteful for scans with selective predicates, joins, aggregations, dictionary lookups, or late materialization. Reducing block sizes improves access granularity, but usually hurts compression effectiveness.

\sparagraph{Dictionary Compression for Strings.} Dictionary encoding is the standard approach for compressing low-cardinality string columns. It stores each distinct string once and replaces column values with integer identifiers. This is highly effective when many values repeat exactly, and it enables efficient query processing over integer codes. However, dictionary encoding does not compress the dictionary strings themselves. For high-cardinality string columns, such as URLs, identifiers, paths, or log messages, many strings are similar but not identical, which limits the effectiveness of pure dictionary encoding. \texttt{DICT\_FSST} addresses this limitation by combining dictionary encoding with \fsst{}: the column is first deduplicated into a dictionary of unique strings, and \fsst{} is then applied to compress the dictionary payload itself~\cite{bruineman2025dictfsst}. This hybrid design exploits both full-string repetition and repeated substrings inside unique dictionary entries, making it a natural integration point for improvements to the underlying \fsst{} compressor.

\sparagraph{FSST.} Fast Static Symbol Table (\fsst{}) was designed specifically for database string columns that require fast decompression and random access~\cite{boncz2020fsst}. \fsst{} builds a static symbol table from a sample of the input and replaces frequent byte sequences of length up to eight with one-byte codes. Bytes that are not represented by a symbol-table entry are emitted as escaped characters: \fsst{} writes a reserved escape code (255) followed by the original byte, so such characters cost two output bytes. 

Since the symbol table is immutable during decompression, each compressed string can be decoded independently. This design allows compressed strings to remain byte strings, simplifying system integration and enabling equality comparisons on compressed data.

The main algorithmic challenge in \fsst{} is constructing a good symbol table. The original algorithm uses an iterative bottom-up procedure: it compresses a sample with the current table, counts symbol and pair frequencies, and constructs the next generation of candidate symbols from high-gain symbols and their concatenations~\cite{boncz2020fsst}. This process addresses some dependencies between overlapping symbols by observing which symbols are actually used during compression. However, \fsst{} still relies on greedy decisions. During encoding, it selects the longest matching symbol at the current position. During table construction, it ranks candidates using local gain estimates. The \fsst{} paper already identifies symbol dependencies and overlapping candidates as central difficulties in symbol-table construction~\cite{boncz2020fsst}. \optfsst{} builds directly on this observation by replacing greedy encoding decisions with a dynamic-programming formulation and by improving candidate selection through additional frequency counting and pruning.

\sparagraph{FSST Extensions.} \fsst{}+ is a recent extension that targets a different source of redundancy: common prefixes across neighboring strings~\cite{alexandre2025fsstplus}. \fsst{}+ groups strings into small blocks, sorts them locally, and uses dynamic programming to select prefix-sharing groups. Shared prefixes are stored once and referenced by suffixes, while the prefix and suffix data are still compressed using \fsst{}. This makes \fsst{}+ complementary to \optfsst{}. \fsst{}+ improves the layout around \fsst{} by exploiting cross-string prefix redundancy, whereas \optfsst{} improves the internal \fsst{} symbol table and encoding decisions. As a result, \optfsst{} can potentially benefit systems that already use \fsst{} as a component, including \texttt{DICT\_FSST} and prefix-oriented extensions such as \fsst{}+.

\sparagraph{GPU-aware FSST.} Recent GPU-enabled query engines have renewed interest in making column encodings and compression schemes accelerator-friendly~\cite{he2022tqp,wu2025blink,coddspeed2026}. In parallel, recent work has shown that GPU-aware encodings and file-format configurations can substantially improve analytical scan performance~\cite{afroozeh2024fastlanes,hapkema2025galp,luo2026gpuparquet}, in particular for string data via a GPU-aware FSST adaptation~\cite{anema2025gpufsst}. \optfsst{}'s DP can also be GPU-accelerated, particularly via NVIDIA Blackwell's recent DPX instructions~\cite{nvidia-blackwell}.

\sparagraph{OnPair.} OnPair is a recent field-level string compressor for in-memory workloads that, like FSST, compresses strings independently and therefore supports fine-grained
random access~\cite{onpair}. In contrast to \fsst{}'s small static
symbol table with one-byte codes and symbols of length up to eight bytes,
OnPair uses a larger dictionary with fixed two-byte token identifiers and builds
the dictionary by incrementally merging frequent adjacent token pairs during a
sequential pass over a sample. Its OnPair16 variant bounds dictionary entries to
16 bytes to make longest-prefix matching and decompression more hardware
friendly. Thus, OnPair explores a different point in the design space: it utilizes pair merging and larger dictionaries to improve compression
effectiveness, whereas \fsst{} prioritizes a compact fixed-size symbol-table, faster
encoding, and simple one-byte symbols.

\sparagraph{Positioning of OptFSST.} \optfsst{} differs from block-based compressors by preserving \fsst{}'s random-access decompression model. It differs from dictionary and front-coding approaches by compressing substrings inside individual string values rather than relying on exact duplicates or lexicographic adjacency. It also differs from \fsst{}+ by optimizing \fsst{}'s symbol selection and encoding process instead of introducing a new prefix-sharing layout. The closest baseline is therefore the original \fsst{} algorithm. \optfsst{} keeps the same high-level abstraction: a static table of byte-sequence symbols, but replaces parts of the greedy construction and encoding pipeline with globally informed decisions.

As a footnote in \fsst{}'s original paper~\cite{boncz2020fsst}, the authors do indeed provide a DP approach as an alternative, but without reporting any performance numbers. This motivated our work, in which we extensively benchmark this approach, optimize the ``severely affected'' encoding, i.e., the inner DP loop, and extend the analysis to \fsstxx{}, which got adopted in CWI's FastLanes~\cite[Table 1]{afroozeh2025fastlanes}.

\section{Approach}
\label{sec:approach}

\optfsst{} improves \fsst{} by modifying only the compression and symbol-table construction pipeline. The decompression format remains unchanged: compressed strings are still byte sequences of symbol codes and escape codes interpreted through a static symbol table. As a result, \optfsst{} preserves \fsst{}'s random-access decompression property while improving compression factors through three contributions: dynamic-programming-based encoding, a third frequency counter, and symbol pruning. We also apply the same ideas to \fsstxx{}, resulting in \optfsstxx{}.

\subsection{Dynamic-Programming-Based Encoding}
\label{subsec:dp-encoding}

\sparagraph{Motivation.} As discussed above, \fsst{}'s longest-match encoder is fast but not optimal for a fixed symbol table. Since each symbol code costs one byte and an escaped
character costs two bytes, fixed-table encoding can be formulated as finding
the minimum-cost segmentation of the input string. \optfsst{} solves this problem
using dynamic programming.

\begin{figure}[t]
    \centering
    \includegraphics[width=\linewidth]{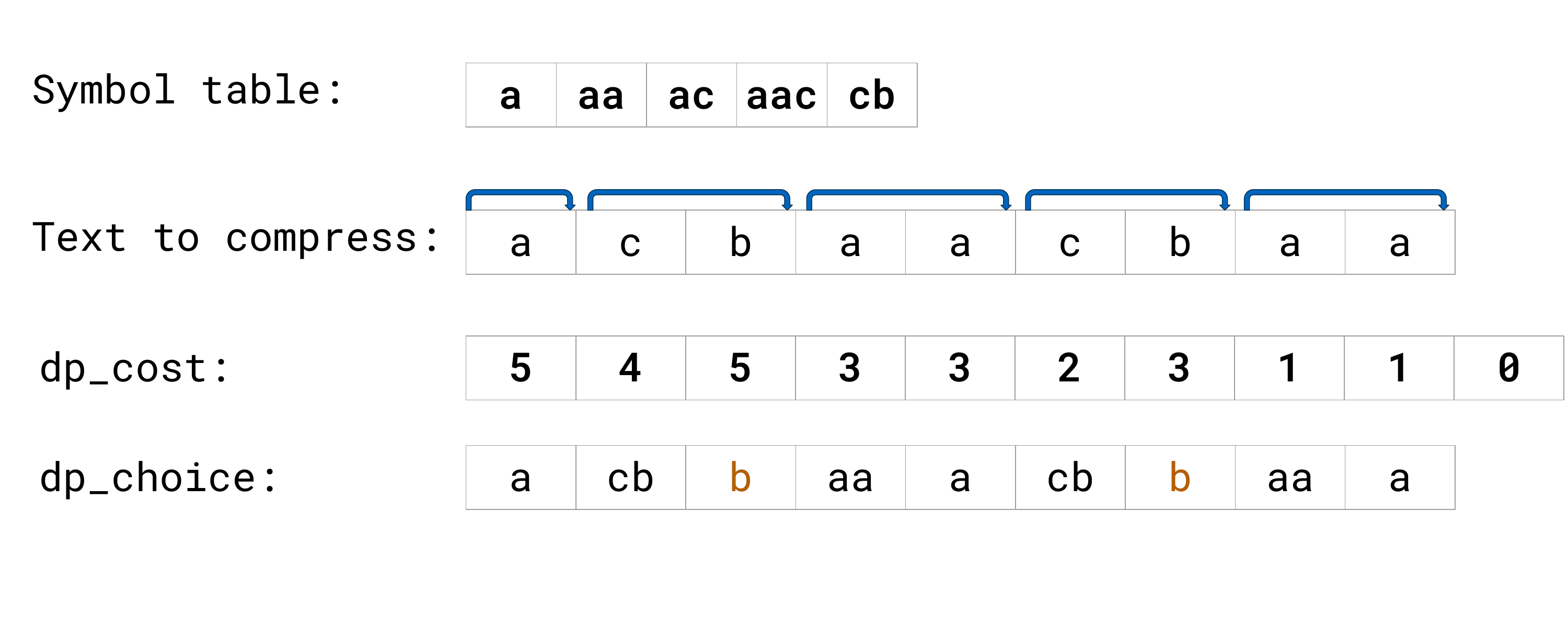}
    \caption{Dynamic-programming-based encoding example. \optfsst{} computes the optimal suffix cost for each input position and follows the selected decisions to obtain a globally better encoding (marked by the arrows) than greedy longest-match selection. Escaped characters are marked in \textcolor{brown}{brown}.}
    \label{fig:dp-compression}
\end{figure}

\sparagraph{Dynamic Programming Formulation.} For an input string $s$ of length $n$, \optfsst{} defines $dp[i]$ as the minimum compressed size of the suffix $s[i:n]$. The base case is $dp[n] = 0$. At each position $i$, the encoder can either escape the current byte, which costs two output bytes, or emit one symbol from the current symbol table $T$ that starts at $i$, which costs one output byte. The recurrence is:
\[
dp[i] =
\min\left(
    2 + dp[i+1],
    1 + \min_{\substack{i < j \leq n \\ s[i:j] \in T}} dp[j]
\right).
\]
In addition to the cost array, \optfsst{} stores an array $\texttt{dp\_choice[i]}$ containing one best decision for position $i$. After filling both arrays from right to left, compression follows the decisions stored in \texttt{dp\_choice}, as shown graphically in Figure~\ref{fig:dp-compression}.

\sparagraph{Trie-based Lookup.} A direct implementation of the recurrence would test all symbols at every input position. \optfsst{} avoids this by storing the current symbol table in a trie. Starting from position $i$, the encoder walks the trie byte by byte for at most the maximum \fsst{} symbol length. If the traversal reaches a terminal node, the represented symbol matches the current input prefix. If no child exists for the next byte, traversal stops immediately. This avoids full-symbol comparisons and limits candidate enumeration to symbols that actually match the input.

\begin{figure}[!t]
\begin{lstlisting}[
  language=C++,
  caption={Core dynamic-programming construction used by \optfsst{}.},
  label={lst:build-dp}
]
void (*@\green{BuildDP}@*)(const uint8_t* data, int n,
             const vector<TrieNode>& trie,
             vector<uint32_t>& dp_cost,
             vector<uint16_t>& dp_choice) {
  dp_cost.assign(n + 1, 0);
  dp_choice.assign(n, 255);  // Default: escape.

  dp_cost[n] = 0;  // Base case.

  for (int i = n - 1; i >= 0; --i) {
    uint32_t best_cost = (*@\red{2 + dp\_cost[i + 1]}@*);
    uint16_t best_code = 255;

    int (*@\pblue{node}@*) = 0;
    uint32_t limit = min(Symbol::maxLength, (uint32_t)(n - i));

    for (uint32_t off = 0; off < limit; ++off) {
      uint8_t byte = data[i + off];
      node = trie[node].child[byte];

      if (node == -1) 
        break;

      int code = trie[node].symbol_code;
      if (code != -1) {
        uint32_t len = off + 1;
        uint32_t cost = 1 + dp_cost[i + len];

        if (cost <= best_cost) {  (*@\label{line:smaller-equal}@*)
          (*@\orange{best\_cost}@*) = cost;
          (*@\purple{best\_code}@*) = code;
        }
      }
    }

    dp_cost[i] = best_cost;
    dp_choice[i] = best_code;
  }
}
\end{lstlisting}
\end{figure}

\sparagraph{Implementation.} Listing~\ref{lst:build-dp} shows the core of the \green{DP} implementation. Each position is initialized with the \red{escape cost}. Then \optfsst{} traverses the \pblue{trie} of symbols starting at the current byte position. Whenever a terminal trie node is reached, the corresponding symbol is a valid candidate and can update the \orange{best cost}. Experiments showed that preferring longer symbols, as a tie-breaker for equal DP costs, yields better overall compression factors; hence the use of ``$\leq$'' in the cost update condition in line~\ref{line:smaller-equal}. The final implementation computes the dynamic-programming cost and the \purple{selected code} in the same pass.

\sparagraph{Use During Table Construction.} \optfsst{} uses the DP encoder both for final compression and during symbol-table construction. \fsst{} builds its symbol table iteratively: in each generation, it compresses a sample using the current table, counts emitted symbols and pairs, and then constructs the next generation from high-gain candidates. Replacing the greedy encoding in this phase changes the frequency signal used for table construction. Counts now reflect symbols that participate in globally good encodings rather than symbols selected only because they are locally longest, improving the feedback loop between compression and symbol selection.

\sparagraph{Traversal Optimizations.} The dynamic-programming step introduces additional compression work, with trie traversal being the main cost. Therefore, \optfsst{} optimizes the \texttt{BuildDP} traversal path using loop unrolling, compiler branch-prediction hints, and small implementation-level refinements. These optimizations do not change the recurrence or the selected encoding; they merely reduce the cost of computing the \texttt{dp\_cost} and \texttt{dp\_choice} arrays. We evaluate the effect of these optimizations in~\secref{subsec:speed-results}.

\sparagraph{Adaptive and Vectorized Variants.} We also evaluated two variants aimed at reducing the DP overhead. First, we tested an adaptive strategy that switches to greedy full-text encoding only when a DP-based compression pass over the training sample exceeds a configurable compressed-size threshold. This did not contribute positively to the compression-factor to compression-speed trade-off compared to always choosing the DP path. 
Second, we tested a SIMD-oriented bucket layout that groups symbols by their
first bytes and vectorizes candidate comparisons within the selected bucket.
In our experiments, this layout was slower than the optimized trie traversal,
because bucket lookup and candidate filtering dominated the benefit of
vectorized comparisons. We therefore retain the trie-based DP implementation.

\subsection{Faster Convergence}
\label{subsec:third-counter}

\sparagraph{Motivation.} \fsst{}'s original construction counts individual symbols through an array \texttt{count1}, and adjacent symbol pairs with a 2D array \texttt{count2}. Pair counts allow the next generation to form longer candidates by concatenating two symbols that occur next to each other. However, longer recurring patterns may require multiple generations before they become visible as candidates. This slows convergence and can prevent useful long symbols from entering the table early enough.

\sparagraph{Counting Triples.} \optfsst{} introduces a third frequency counter, \texttt{count3}, for triples of consecutive emitted symbols.\footnote{The motivation is that a Markov model, as (implicitly) used in the pair counting of FSST (\texttt{count2}), can better capture the distribution when conditioned on a longer context; in this case, of length two.} During sample compression, \optfsst{} records patterns of the form $(a,b,c)$ in addition to the original single-symbol and pair counts. Candidate generation can then concatenate three adjacent symbols directly, as long as the resulting symbol respects the maximum symbol length. This exposes longer recurring patterns earlier than pairwise growth alone. We also experimented with extending this idea to higher-order counters, beginning from \texttt{count4}, but they did not yield significant additional improvements in our experiments.

\sparagraph{Compact Representation.} A dense three-dimensional counter would be too large. Since only triples that actually occur in the compressed sample are relevant, \optfsst{} stores \texttt{count3} sparsely in a hash map. Each \fsst{} code fits into nine bits, so three codes can be packed into a 32-bit integer.

\sparagraph{Candidate Generation.} During the new symbol table population of each generation, \optfsst{} considers the original \fsst{} candidates from \texttt{count1} and \texttt{count2}, and additionally iterates over the non-zero entries of \texttt{count3}. For every observed triple, the three corresponding symbols are concatenated and inserted into the candidate heap if the resulting symbol is short enough. The gain is computed in the same style as \fsst{}, multiplying the candidate length by the observed frequency. Thus, \optfsst{} preserves \fsst{}'s generation-based construction, but extends the candidate space with longer combinations discovered by the third counter.

\subsection{Symbol Pruning}
\label{subsec:symbol-pruning}

\begin{figure}[t]
    \centering
    \includegraphics[width=\linewidth]{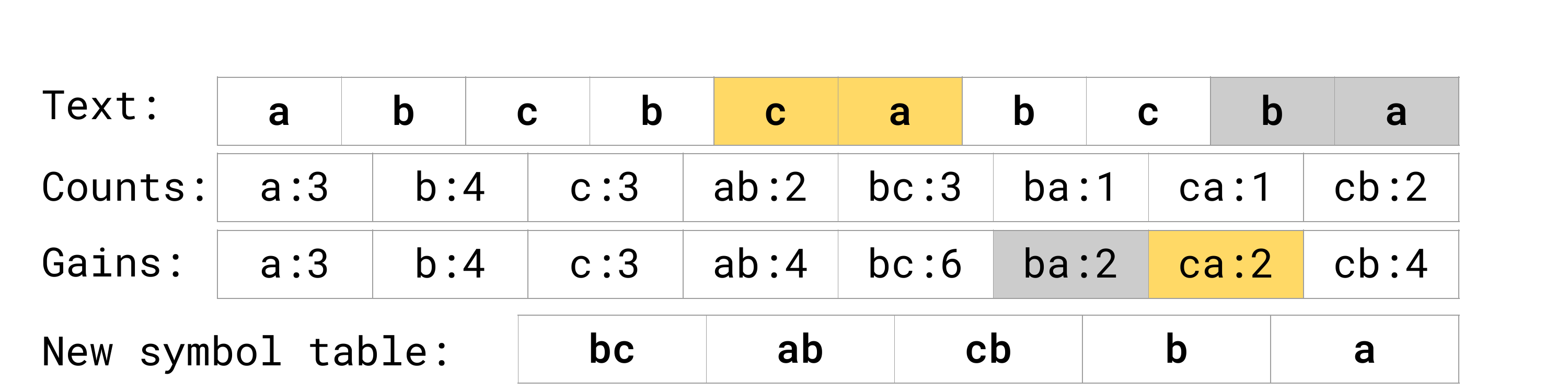}
    \caption{ Example of symbol selection without pruning. The symbols ``ba`` and ``ca`` are not chosen, because symbols ``a`` and ``b`` had higher static gains. \optfsst{} handles this case optimally.}
    \label{fig:symbol-pruning}
\end{figure}

\sparagraph{Motivation.} \fsst{} ranks candidate symbols by static gain, which is based on frequency and symbol length. This estimate can be inaccurate when symbols overlap. If a larger symbol is selected, smaller symbols contained in it may no longer be emitted often in the next generation. Without correction, these smaller symbols can still appear valuable and occupy entries in the limited symbol table. \optfsst{} introduces pruning to partially correct this overestimation during candidate selection.

\sparagraph{Example w/o Pruning.} Figure~\ref{fig:symbol-pruning} illustrates the issue for a table with capacity five. Static gain selects the high-count one-byte symbols ``a'' and
``b'', although replacing them with the two-byte symbols ``ba'' and ``ca''
would cover the text more effectively. The example shows that gains computed
before candidate selection can become stale once overlapping longer symbols
enter the table.

\sparagraph{Pruning Logic.} When \optfsst{} selects a concatenated symbol, it reduces the counts of the smaller symbols and sub-combinations used to form it. For example, if a selected candidate corresponds to the concatenation $XYZ$, the counts of $X$, $Y$, $Z$, $XY$, and $YZ$ are reduced. Their gains are then recomputed before they can be selected. Candidates whose corrected gain becomes non-positive are discarded. This gives other candidates the opportunity to enter the table and reduces redundancy among selected symbols.

Returning to the example of Figure~\ref{fig:symbol-pruning}, \optfsst{} would reduce the counts of the symbols ``b'' and ``c'' by the count of ``bc'', which is 3, right after ``bc'' is added to the new symbol table. In this way, the symbol ``c'' will be immediately discarded as its count reached 0, but the symbol ``b'' still remains with the count 1 and gets reinserted into the candidate heap. The symbol ``b'' will be discarded after adding ``ab'' and the symbol ``a'' will be discarded after adding the fourth symbol. Then we would have all 2-byte symbols in the new symbol table, which would compress the given text optimally.   

\sparagraph{Implementation.} Listing~\ref{lst:pruning} shows a simplified version of the pruning step. \orange{Candidate} gains stored in the heap may become stale after earlier pruning operations, so every popped candidate is validated against its current \pblue{corrected count} after being \purple{reconstructed} from its individual parts through concatenation. If the candidate is still valid, it is inserted into the next symbol table and its parts are \green{pruned}.

\begin{figure}[!t]
\begin{lstlisting}[
  language=C++,
  caption={Simplified candidate selection with pruning.},
  label={lst:pruning}
]
struct (*@\orange{Candidate}@*) {
  uint16_t a;
  uint16_t b;
  uint16_t c;
  uint8_t arity;  // 1, 2, or 3.
  uint32_t gain;
};

while (next_table.size() < max_symbols && !heap.empty()) {
  Candidate (*@\orange{candidate}@*) = heap.PopMax();

  Symbol symbol = (*@\purple{Reconstruct}@*)(candidate);
  uint32_t current_count = (*@\pblue{CorrectedCount}@*)(candidate);
  uint32_t current_gain = current_count * symbol.length();

  // Ignore stale heap entries.
  if (current_gain != candidate.gain || current_gain == 0) 
    continue;

  if (next_table.Contains(symbol)) 
    continue;

  next_table.Insert(symbol);

  // Correct counts of covered smaller candidates.
  if (candidate.arity >= 2) {
    (*@\green{PruneSymbol}@*)(candidate.a, current_count);
    PruneSymbol(candidate.b, current_count);
  }

  if (candidate.arity == 3) {
    PruneSymbol(candidate.c, current_count);
    (*@\green{PrunePair}@*)(candidate.a, candidate.b, current_count);
    PrunePair(candidate.b, candidate.c, current_count);
  }
}
\end{lstlisting}
\end{figure}

\begin{figure*}[t]
    \centering

    \begin{subfigure}[t]{0.49\textwidth}
        \centering
        \includegraphics[width=\linewidth]{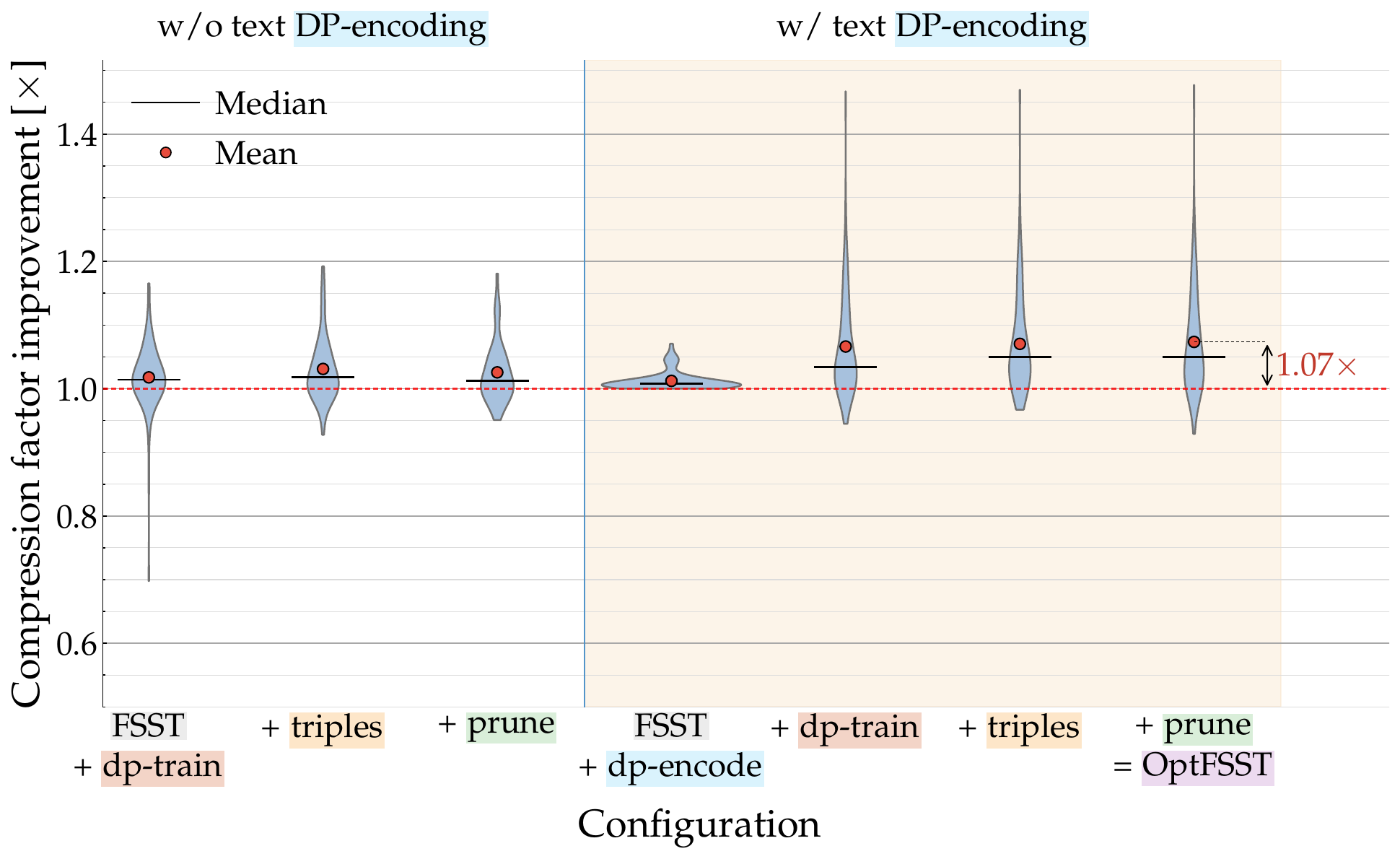}
        \caption{\optfsst{}}
        \label{fig:cf-improvement-optfsst}
    \end{subfigure}
    \hfill
    \begin{subfigure}[t]{0.49\textwidth}
        \centering
        \includegraphics[width=\linewidth]{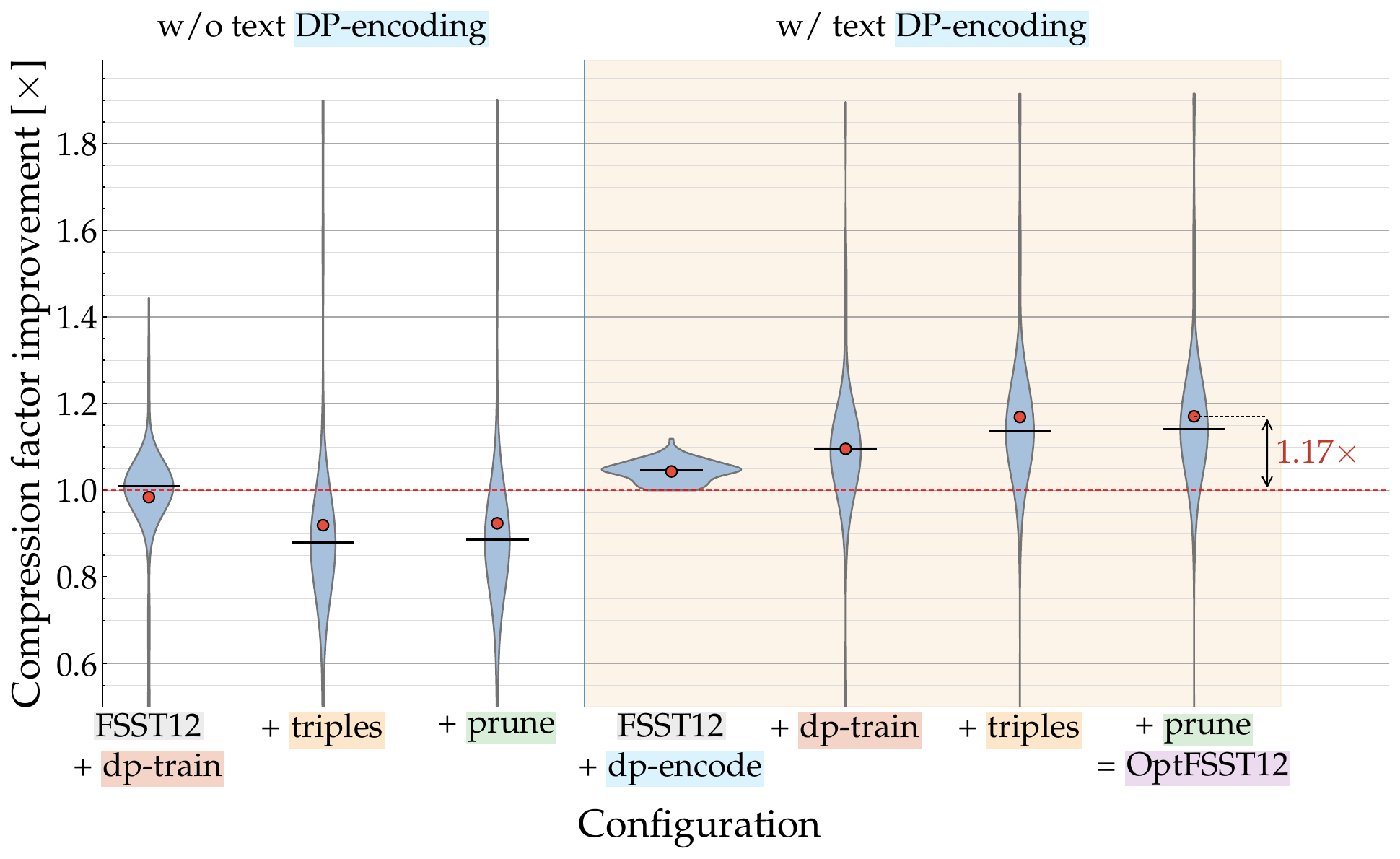}
        \caption{\optfsstxx{}}
        \label{fig:cf-improvement-optfsst12}
    \end{subfigure}

    \caption{Compression-factor improvement over the corresponding baseline. The dashed red line marks parity with \fsst{} or \fsstxx{}. Each box shows the distribution across benchmark columns, and the red point denotes the mean.}
    \label{fig:cf-improvement}
\end{figure*}

\subsection{OptFSST12}
\label{subsec:optfsst12}

We also apply the same three contributions to \fsstxx{}, resulting in \optfsstxx{}. \fsstxx{} provides a larger symbol space and uses a different symbol storage design, but the main optimization problem remains the same: the compressor must decide which symbols to place in the table and how to encode strings with them. \optfsstxx{} therefore uses the dynamic-programming encoder, the third frequency counter, and pruning in the \fsstxx{} setting.

\sparagraph{Preserved Properties.} Both \optfsst{} and \optfsstxx{} preserve the central properties of \fsst{}. The symbol table remains static during decompression, each compressed string can still be decompressed independently, and the decompression path is unchanged. The additional computation is concentrated in compression and table construction, where more informed optimization can improve compression factors. This design deliberately trades additional compression effort for better compressed size while keeping \fsst{}'s fast random-access decompression model intact.

\section{Optimal Symbol Table Selection}
\label{sec:np-hardness}
FSST uses a greedy strategy in both the symbol-table construction and encoding phase. While \optfsst{} guarantees the optimal encoding with respect to a \emph{given} table, finding the symbol table is still suboptimal (even though we are using DP as a subroutine to evaluate the quality of the current table). A natural question is whether the problem is inherently hard, thereby justifying the proposed heuristics, or whether we simply have not yet found an efficient algorithm. In this work, we show the problem of selecting an optimal symbol table is indeed NP-hard.\footnote{Our reduction uses an alphabet whose size grows with the input. If both $\ell$ and the alphabet size are fixed, then the problem is solvable in polynomial time; see Appendix~\ref{appendix:STS_l_with_bounded_alphabet_size_in_P}.
If instead the dictionary size is fixed, the problem is also solvable in polynomial time; compare~\cite[Sec.~4.1]{boncz2020fsst}.}

We formalize the problem in the following, emulating the setting we have in FSST:

\begin{problem}[FSST Symbol Table Selection]
Fix an integer \(\ell \ge 1\). An instance of \(\STS_{\ell}\) consists of:
\begin{itemize}
    \item a finite multiset \(\mathcal{T}\) of strings over an alphabet \(\Sigma\),
    \item an integer \(K \ge 0\), and
    \item a target bound \(B \ge 0\).
\end{itemize}
A dictionary \(D\) is a set of strings over \(\Sigma\) of length at most \(\ell\), with \(|D| \le K\). A \emph{\(D\)-encoding} of a string \(s \in \mathcal{T}\) is a partition of \(s\) into consecutive phrases, each of which is either:
\begin{itemize}
    \item a literal byte \(c \in \Sigma\), encoded at cost \(2\) (escape byte plus \(c\)), or
    \item an occurrence of a dictionary symbol \(w \in D\), encoded at cost \(1\).
\end{itemize}
The cost of an encoding is the sum of phrase costs, and \(\cost_D(s)\) denotes the minimum cost of any \(D\)-parsing of \(s\). Finally,
\[
\cost_D(\mathcal{T}) \;:=\; \sum_{s \in \mathcal{T}} \cost_D(s).
\]

The decision problem asks whether there exists a dictionary \(D\) with \(|D| \le K\) such that
\[
\cost_D(\mathcal{T}) \le B.
\]
\end{problem}

We show that $\STS_\ell$, the problem of selecting symbols of length at most $\ell$ (including escaping single bytes), is NP-hard:

\begin{restatable}{theorem}{MAINTHM}
	STS$_\ell$ is NP-hard for every fixed $\ell\geq 2$.
    \label{thm:np-hardness}
\end{restatable}

We refer the reader to Appendix~\ref{appendix:np-hardness} for the proof of Theorem~\ref{thm:np-hardness}. Note that, unlike the tokenization problem found in modern language models, the STS$_\ell$ problem fixes the length $\ell$ of a token to be considered; in FSST, a symbol's length is bounded by 8 bytes. Hence, the recent NP-hardness proofs for the tokenization problem, as in Ref.~\cite{token-hard, token-hard-bounded}, are not directly applicable. Our proof thus contributes to the hardness landscape for the tokenization problem.

\section{Evaluation}
\label{sec:evaluation}

This section evaluates \optfsst{} and \optfsstxx{} with respect to compression-factor improvements and the effect on compression and decompression speed. We first describe the benchmark datasets and methodology. We then present an ablation study that isolates the contribution of each optimization. Finally, we analyze the runtime impact of the proposed changes.

\sparagraph{Setup.} We conduct the experiments on a single node Intel\textsuperscript{\textregistered} Xeon\textsuperscript{\textregistered} Gold 5318Y CPU (24 cores, 48 hyper-threads). The machine is equipped with 128GB DDR4 main memory and runs Ubuntu 24.04. All experiments are run single-threaded.

\subsection{Benchmark Datasets}
\label{subsec:evaluation-datasets}

We evaluate \optfsst{} on heterogeneous real-world string datasets. The goal is to cover different classes of string columns, including natural language text, URLs, identifiers, semi-structured strings, and short categorical strings. We use the following datasets:

\begin{itemize}
  \item \textbf{dbtext} (23 files): the database string corpus introduced with \fsst{}, containing representative string columns from database workloads~\cite{cwida2020fsst}.
  \item \textbf{NextiaJD} (22 files): a collection of real-world datasets used for data discovery experiments, containing heterogeneous string columns such as URLs, comments, and structured identifiers~\cite{nextiajd2021}.
  \item \textbf{Public BI Benchmark} (31 files): a business-intelligence benchmark containing realistic analytical columns, including textual and categorical attributes~\cite{ghita2019publicbi}.
  \item \textbf{CyclicJoinBench} (13 files): a benchmark for join-heavy analytical workloads that also contains string-typed columns~\cite{cyclicjoinbench}.
  \item \textbf{ClickBench} (3 files): a clickstream analytics benchmark with web-oriented string columns such as URLs, referrers, user-agent strings, titles, and categorical attributes~\cite{clickbench2022}.
\end{itemize}

We focus on string columns for which \fsst{} is meaningful. Columns that are compressed better using dictionary encoding are filtered out. We also require a sufficiently large number of strings per column (at least 1000) to obtain stable compression and runtime measurements. Each configuration is evaluated on the same set of columns and compared against the corresponding baseline, \fsst{} or \fsstxx{}.

We report compression-factor improvement as a multiplicative ratio over the baseline:
\[
    \text{improvement} =
    \frac{\text{compression factor of variant}}
         {\text{compression factor of baseline}}.
\]
Values above $1.0$ therefore indicate an improvement over the baseline, while values below $1.0$ indicate a regression. Runtime results are reported in MB/s. For compression speed, we measure only the final corpus compression after the symbol table has already been constructed. Table construction speed is measured separately, since the \fsst{} SIMD version does not affect this phase.

\begin{figure*}[t]
    \centering

    \begin{subfigure}[t]{0.32\textwidth}
        \centering
        \includegraphics[width=\linewidth]{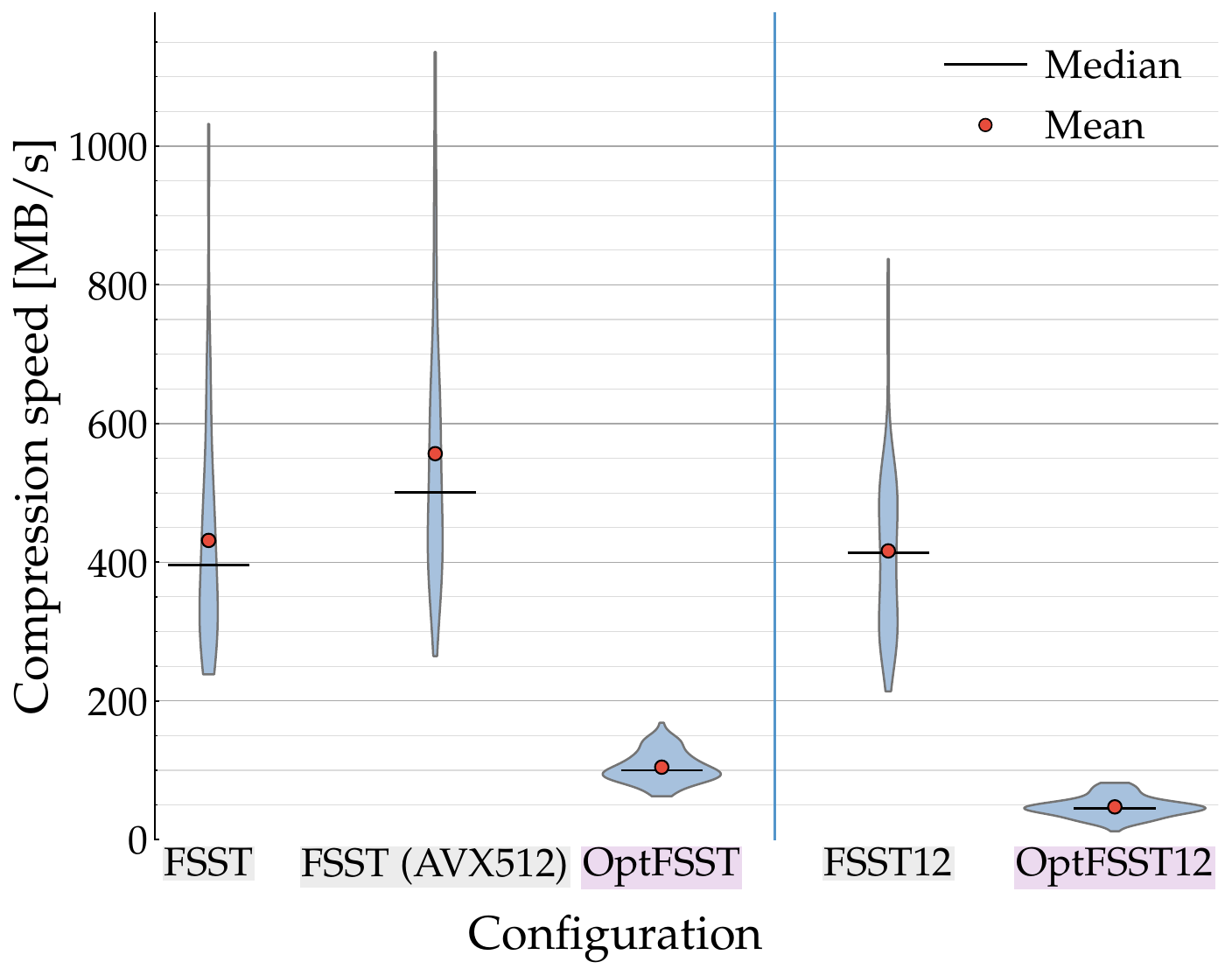}
        \caption{Final corpus compression speed [MB/s]}
        \label{fig:compression-speed}
    \end{subfigure}
    \hfill
    \begin{subfigure}[t]{0.32\textwidth}
        \centering
        \includegraphics[
            width=\linewidth,
        ]{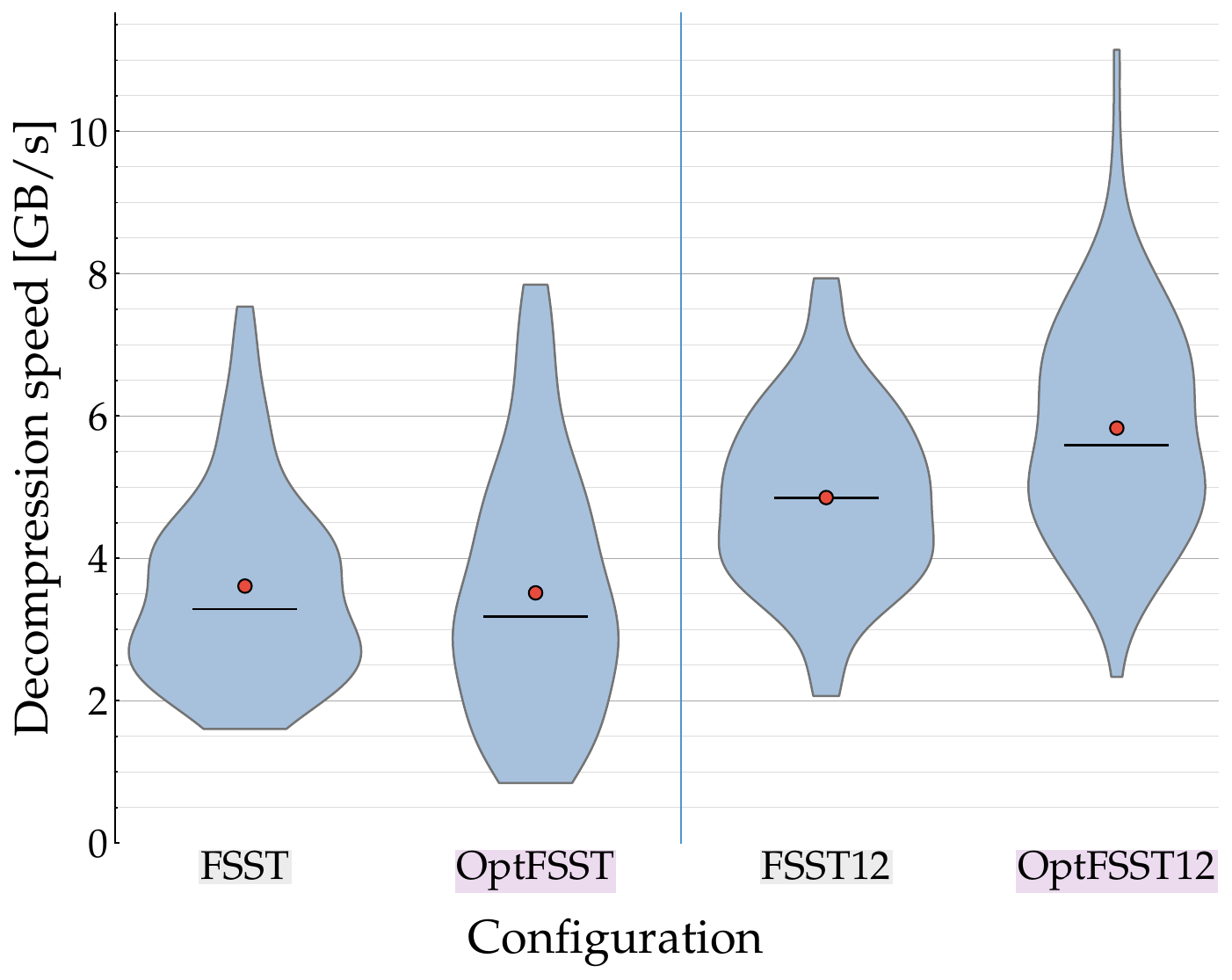}
        \caption{Decompression speed [GB/s]}
        \label{fig:decompression-speed}
    \end{subfigure}
    \hfill
    \begin{subfigure}[t]{0.32\textwidth}
        \centering
        \includegraphics[width=\linewidth]{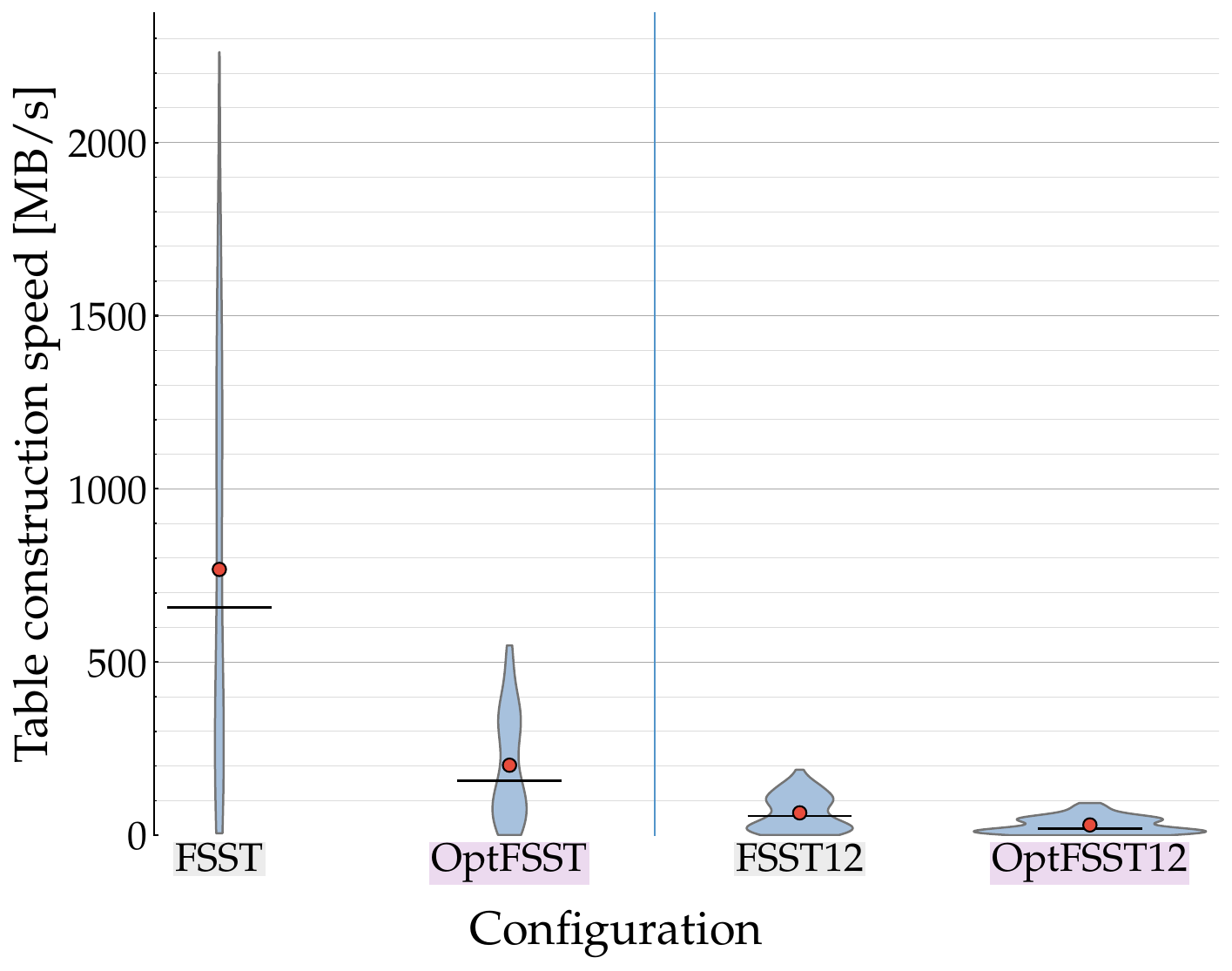}
        \caption{Symbol-table construction speed [MB/s]}
        \label{fig:table-construction-speed}
    \end{subfigure}

    \caption{Runtime comparison. Compression speed measures only final corpus compression after the symbol table has been constructed. Table construction is reported separately because \optfsst{} performs additional optimization during this phase.}
    \label{fig:speed-results}
\end{figure*}

\subsection{Compression Factor Improvement}
\label{subsec:compression-factor-improvement}

\sparagraph{OptFSST.}
Figure~\ref{fig:cf-improvement-optfsst} shows the ablation study for \optfsst{}. The first configuration, \emph{\gray{FSST} + \red{dp-train}}, uses dynamic programming during symbol-table construction. This already improves compression factors by $1.02\times$ on average. The effect is moderate because only the training signal changes: the final corpus is still encoded with the baseline greedy encoder. The \emph{+ \orange{triples}} configuration adds the third frequency counter and reaches an average improvement of $1.03\times$. This shows that exposing longer recurring symbol combinations earlier during table construction helps the algorithm discover more useful symbols.

The \emph{+ \green{prune}} configuration adds pruning during candidate selection. In terms of the aggregate mean, pruning has no noticeable effect on the final compression factor. However, the aggregate mean hides the fact that pruning is beneficial for a majority of columns.
For \optfsst{}, pruning improves the final compression factor for $56.5\%$ of the files, with the best case reaching a $1.05\times$ improvement for the sparse column \texttt{HashTags\_1::null\_21} from the \texttt{Public BI Benchmark} dataset. Without text DP \pblue{encoding}, individual columns still benefit substantially, with the best case reaching $1.113\times$. Thus, \green{pruning} is not the main source of the headline average improvement, but it consistently helps in the cases where overlapping candidates would otherwise lead to stale or overestimated gains.

The \emph{\gray{FSST} + \pblue{dp-encode}} configuration applies dynamic programming to the final corpus \pblue{encoding} and reaches $1.01\times$ on average. This confirms that even for a fixed symbol table, greedy longest-match encoding can leave compression opportunities unused.

The final \purple{\optfsst{}} configuration combines DP \red{training}, \orange{triple} counting, dynamic-programming \pblue{encoding}, and \green{pruning}. It reaches an average compression-factor improvement of $7.3\%$ over \fsst{}. \optfsst{} improves the compression factor over \fsst{} on $82/92$ benchmark columns ($89\%$). The main aggregate gain comes from the two dynamic-programming uses and from the additional triple counter.

The violin plot also shows that the improvement is not uniform across all columns. Some columns remain close to parity, while others benefit substantially more, with outliers reaching up to $47.7\%$ improvement, which is observed in a column with short strings of length 2. This variation is expected because \optfsst{} mainly improves cases where symbol overlap and greedy encoding decisions matter. Columns that are already well represented by \fsst{}'s greedy table construction leave less room for improvement.

\sparagraph{OptFSST12.}
Figure~\ref{fig:cf-improvement-optfsst12} repeats the ablation study for \fsstxx{}. The results show that the same ideas remain effective even when the baseline has a larger symbol space. However, the individual contributions behave differently. The \emph{\gray{FSST12} + \red{dp-train}} configuration is slightly below parity on average, at $0.98\times$. This indicates that changing only the training encoder can sometimes disturb \fsstxx{}'s symbol selection when the final encoding remains unchanged. In contrast, the \emph{+ \orange{triples}} configuration improves the average to $1.04\times$, showing that longer candidate discovery is still useful in the larger-symbol-space setting.

For \fsstxx{}, \green{pruning} again has no noticeable aggregate effect on the final compression factor. As with \optfsst{}, this neutral mean should not be interpreted as pruning being unused. Pruning improves the final compression factor for $54.3\%$ of the files, and the best final-encoding case reaches $1.03\times$. Without text DP \pblue{encoding}, the effect is even more visible for \fsstxx{}: \green{pruning} improves the compression factor for $67.4\%$ of the files, with a best case of $1.19\times$ on the column \texttt{c\_name} from the \texttt{dbtext} dataset. This suggests that the larger symbol space creates more opportunities for overlapping candidates to compete, making stale gain correction useful even when the final aggregate mean remains unchanged. We therefore keep \green{pruning} in the final configuration as a lightweight correction for stale gain estimates caused by overlapping candidates, rather than as a mechanism that is expected to increase the aggregate mean on every workload.

The final \purple{\optfsstxx{}} configuration combines all previously mentioned contributions and achieves an average compression-factor improvement of $1.17\times$ over \gray{\fsstxx{}}. \optfsstxx{} improves over \fsstxx{} on $87/92$ columns ($95\%$). These numbers are larger than the relative improvement observed for \optfsst{} over \fsst{}, suggesting that \fsstxx{}'s larger symbol space creates more opportunities for improved candidate selection.

Overall, the ablation study shows that dynamic programming and triple counting are the main sources of compression-factor improvement. Dynamic-programming encoding improves the final segmentation decisions, while the third counter exposes longer recurring candidates earlier during table construction. Pruning serves as a correction mechanism that proved useful for most files.

\begin{table}[t]
\centering
\small
\caption{Runtime ablation. Values are arithmetic-mean slowdowns ($\times$) over the
corresponding \fsst{}/\fsstxx{} baseline. E2E includes both table construction and corpus encoding.}
\label{tab:runtime-ablation}
\begin{tabular}{lrrrr}
\toprule
& \multicolumn{2}{c}{\optfsst{}} & \multicolumn{2}{c}{\optfsstxx{}} \\
\cmidrule(lr){2-3}\cmidrule(lr){4-5}
Configuration & Symbol table & E2E & Symbol table & E2E \\
\midrule
Baseline          & 1.00 & 1.00 & 1.00 & 1.00 \\
+ \red{DP train}  & 2.49 & 1.68 & 1.51 & 1.45 \\
+ \orange{triples}& 4.02 & 2.36 & 2.22 & 2.07 \\
+ \green{prune}   & 4.07 & 2.38 & 2.28 & 2.12 \\
Baseline + \pblue{DP encode} & n/a & 2.33 & n/a & 1.84 \\
All configurations & 4.07 & 3.82 & 2.16 & 3.02 \\
\bottomrule
\end{tabular}
\end{table}

\begin{table*}[t]
    \centering
    \caption{Comparison of FSST-style and \pink{block-based} compressors.}
    \label{tab:block-compressors-combined}
    \begin{tabular}{l|rrrr|rrrr|rrrr}
        \toprule
        \multirow{3}{*}{Metric}
        & \multicolumn{4}{c|}{Compression factor [$\times$]}
        & \multicolumn{4}{c|}{Compression speed [MB/s]}
        & \multicolumn{4}{c}{Decompression speed [GB/s]} \\
        \cmidrule(lr){2-5}
        \cmidrule(lr){6-9}
        \cmidrule(lr){10-13}
        & min & mean & median & max
        & min & mean & median & max
        & min & mean & median & max \\
        \midrule

        \gray{FSST}
        & 1.17 & 2.68  & 2.33 & 6.54
        & 238.36 & 431.93 & 395.55 & 1032.01
        & 1.55 & 3.60 & 3.28 & 7.22 \\
        
        \gray{FSST (AVX512)}
        & 1.17 & 2.68  & 2.33 & 6.54
        & 264.66 & 556.56 & 500.69 & 1135.60 
        & 1.55 & 3.60 & 3.28 & 7.22 \\
        
        \purple{OptFSST}
        & 1.21 & 2.93  & 2.48 & 7.25
        & 61.28  & 104.93 & 99.30  & 166.78
        & 0.85 & 3.54 & 3.18 & 7.96 \\
        
        \midrule
        
        \gray{FSST12}
        & 0.97 & 2.29  & 2.31 & 4.46
        & 213.84 & 416.77 & 409.77 & 849.51
        & 1.99 & 4.91 & 4.89 & 8.40 \\
        
        \purple{OptFSST12}
        & 0.39 & 2.72  & 2.60 & 5.13
        & 15.23  & 47.48  & 45.59  & 82.22
        & 2.55 & 5.86 & 5.87 & 9.16 \\

        \midrule
        
        \pink{Snappy}
        & 1.00 & 3.62  & 2.45 & 19.24
        & 49.69  & 272.59 & 244.40 & 714.27
        & 0.11 & 0.57 & 0.48 & 1.71 \\
        
        \pink{LZ4}
        & 1.00 & 5.81  & 2.36 & 87.63
        & 76.33  & 448.88 & 335.26 & 1989.84
        & 0.50 & 1.61 & 1.26 & 6.00 \\
        
        \pink{ZSTD}
        & 1.55 & 10.83 & 4.02 & 174.34
        & 14.38  & 164.18 & 137.29 & 681.75
        & 0.07 & 0.61 & 0.50 & 2.75 \\
        
        \bottomrule
    \end{tabular}
\end{table*}

\subsection{Effect on (De)Compression Speed}
\label{subsec:speed-results}

\sparagraph{Runtime Ablation.} Table~\ref{tab:runtime-ablation} complements the compression-factor ablation in \secref{subsec:compression-factor-improvement}
by separating table-construction overhead from end-to-end compression time; we report the slowdown over the corresponding baseline. Note that the rows are cumulative for the training-side variants, i.e., \emph{+ \orange{triples}}
includes DP-based \red{training}, and \emph{+ \green{prune}} includes both DP-based training
and triple counting. In the case of full-text \pblue{DP encoding} only, the table-construction time remains unchanged. We observe that, for table construction, DP-based \red{training} and \orange{triple} counting are the main sources of overhead, while \green{pruning} has no noticeable effect. The relative construction overhead is smaller for \fsstxx{} because its baseline table construction is already more expensive due to the larger symbol space.

\sparagraph{\optfsst{}.} Next, we analyze the throughput of \optfsst{}'s compression, decompression, and symbol-table construction.

\paragraph{Compression.} Figure~\ref{fig:compression-speed} shows that \purple{\optfsst{}} reduces final corpus compression speed compared with \gray{\fsst{}}. \fsst{} compresses at $431.9$~MB/s on average, while \optfsst{} reaches $104.9$~MB/s.\footnote{Note that a ``natural'' slowdown would be 8x, as the inner DP loop has to run on all 8 possible lengths. However, the trie, as explained in~\secref{subsec:dp-encoding}, allows exiting non-matching paths faster on average.} The \fsst{} SIMD version is around $1.3\times$ faster than the regular \fsst{}, which increases the gap. This slowdown is expected because \purple{\optfsst{}} replaces the branchless greedy encoder with a dynamic programming \pblue{encoder} that evaluates multiple candidate symbols per input position. Notably, the optimized trie traversal, including loop unrolling and compiler branch-prediction hints, reduces this overhead, as shown in Figure~\ref{fig:trie-opt}: the corpus encoding speed is improved by these optimizations by $1.03\times$ and $1.06\times$ on average in \purple{\optfsst{}} and \purple{\optfsstxx{}}, respectively. This helps further close the gap to FSST's compression performance.

\begin{figure}[t]
    \centering
    \includegraphics[width=0.75\linewidth]{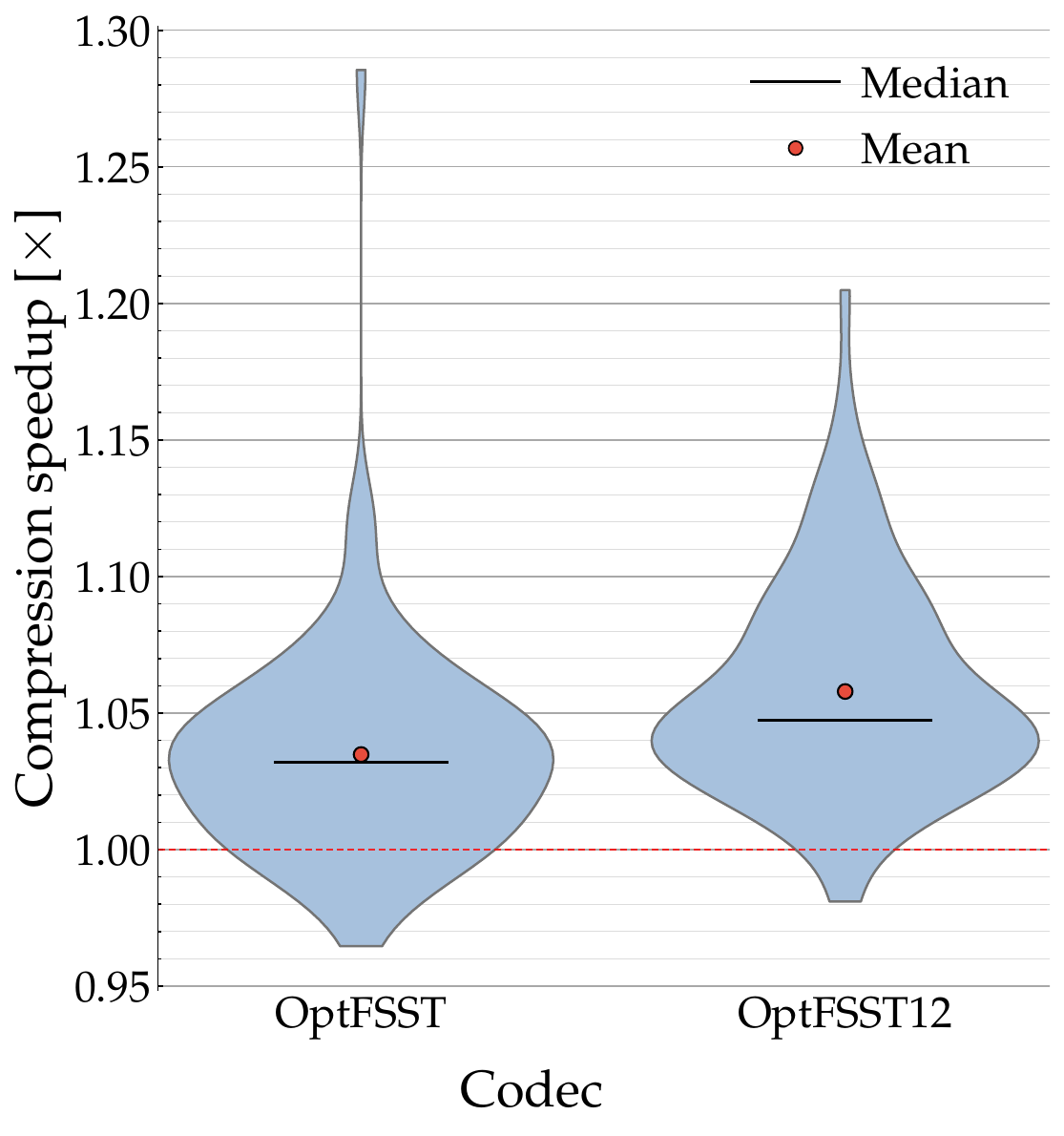}
    \caption{Speedup of corpus compression after the trie traversal optimizations---loop unrolling and branch-prediction hints---in both \optfsst{} and \optfsstxx{}.}
    \label{fig:trie-opt}
\end{figure}

\paragraph{Decompression.} Figure~\ref{fig:decompression-speed} shows that decompression is essentially unaffected. \fsst{} reaches $3.6$~GB/s on average, while \optfsst{} reaches $3.5$~GB/s. Since \optfsst{} does not change the compressed representation or the decompression algorithm, decompression remains a simple symbol-table lookup per emitted code. Small differences in the measured decompression speed are therefore mainly caused by memory behavior and branch prediction, as the placement of escape codes in the compressed text affects the behavior of the decompression function in \fsst{}.

\paragraph{Table Construction.} Figure~\ref{fig:table-construction-speed} shows that during symbol-table construction, baseline \fsst{} constructs tables at $766.8$~MB/s on average. \optfsst{} constructs tables at $200.7$~MB/s. As shown in Table~\ref{tab:runtime-ablation}, the decrease is caused mainly by dynamic-programming \red{training} and \orange{triple} counting. \green{Pruning} has no noticeable effect on construction time, since it operates on the candidate heap after frequency counting. Nevertheless, table construction remains practical, especially because it is performed once per block, row group, or compression unit, while decompression is usually on the critical path for query execution.

\sparagraph{OptFSST12.} Table construction for \gray{\fsstxx{}} is substantially slower than \fsst{}, reaching $64.4$~MB/s on average in our measurements. \purple{\optfsstxx{}} further reduces table construction speed to $31.6$~MB/s. This overhead is expected because \optfsstxx{} applies the same additional optimization steps as \optfsst{}, but in the setting of a larger symbol space. The larger candidate space increases the cost of construction, and the dynamic-programming and pruning steps add further work.

For final compression speed, \purple{\optfsstxx{}} follows the same trend as \optfsst{}: the dynamic-programming \pblue{encoder} is slower than the greedy baseline because it computes globally optimal suffix costs instead of emitting the first longest match. However, this cost is paid during compression, not decompression. Unlike \optfsst{}, \purple{\optfsstxx{}} noticeably improved the average decompression speed of \gray{\fsstxx{}} by $1.2\times$. This is due to the higher improvement of the compression factor, and thus smaller compressed size, combined with the simpler decompression design of \fsstxx{} that does not contain branching caused by escape codes like in \fsst{}. Thus, \optfsstxx{} trades additional compression and construction effort for improved compression factors and decompression speed.

\sparagraph{Trade-off and Hardware Implications.} The runtime results quantify the intended trade-off: \optfsst{} and \optfsstxx{} spend additional CPU time during compression and table construction to produce a more compact representation that is later scanned, cached, and decoded. This is attractive in analytical systems because compression is typically paid during
ingest, compaction, or row-group creation, whereas compressed columns may be
read many times during query processing. On modern CPUs, analytical scans are
often constrained by memory traffic and cache residency. Reducing the string
payload can therefore improve the amount of useful data per cache line and
lower memory-bandwidth pressure. The branch-heavy DP work is placed on the less
frequent compression side, while the frequently executed decompression side
continues to use the simple static-table decoding path. 

\subsection{Comparison with Block-Based Compressors}
\label{subsec:block-compressors}

\sparagraph{Compression Factor.} Table~\ref{tab:block-compressors-combined} compares the FSST-based compressors with the block-based compressors Snappy, LZ4, and ZSTD. The block-based compressors achieve higher compression factors on average because they can exploit redundancy across larger byte ranges. Snappy reaches $3.62\times$ on average, which is $1.2\times$ higher than \optfsst{}, whereas LZ4 and ZSTD achieve much larger average compression factors, which are $1.9\times$ and $3.6\times$ higher than that of \optfsst{}, respectively. This illustrates the main design trade-off: block-based compressors obtain better compression by using a larger compression context, while FSST-style compressors sacrifice some compression effectiveness to preserve fine-grained random access and independent decompression of individual strings.

The minimum compression factor of \optfsstxx{}, 0.39$\times$, is an outlier on
the column \texttt{place/name} from the \texttt{CyclicJoinBench (SNB1)} dataset. This column is small (13.6~KB) and consists almost entirely of distinct (only one duplicate) short strings (8.2 bytes on average). DP training fills the entire 4096-entry symbol table: Its serialization alone exceeds twice the input size and dominates the reported compression factor, while \fsstxx{} constructs a smaller table on the same column.
As a result, \optfsstxx{} can lose on this column even though its mean and median compression factors remain higher than \fsstxx{}.

\sparagraph{Compression and Decompression Speed.} The speed results in Table~\ref{tab:block-compressors-combined} show the opposite side of this trade-off. During compression, the AVX512 version of \fsst{} is the fastest method on average, followed closely by LZ4 and \fsstxx{}. In contrast, \optfsst{} and \optfsstxx{} are slower due to their dynamic-programming encoder. Compared with LZ4, \optfsst{} is $4.2\times$ slower and \optfsstxx{} is $9.4\times$ slower on average. The improved compression factors of \optfsst{}(12) therefore come at a clear compression-time cost.

However, decompression shows why the FSST-based design remains attractive for analytical workloads. On average, \optfsstxx{} decompresses at $5.86$~GB/s, which is $3.6\times$ faster than LZ4, $10.2\times$ faster than Snappy, and $9.6\times$ faster than ZSTD.

Overall, the comparison shows that \optfsst{} and \optfsstxx{} occupy a different point in the design space than general-purpose block compressors. They do not compete with ZSTD or LZ4 on maximum compression factor, and they are slower during compression than the fastest baselines. Instead, they improve the compression effectiveness of FSST-style random-access string compression, thereby improving its high decompression throughput.
\section{Conclusion \& Future Work}
\label{sec:conclusion}
This paper presented \optfsst{}, an optimized \fsst{} variant that improves
compression effectiveness while preserving \fsst{}'s static-table format,
independent decompression, and random-access interface. The key idea is to
spend additional work during compression and table construction, through
DP-based encoding, triple counting, and pruning, while leaving the
decompression-side abstraction unchanged. Our evaluation shows that this
trade-off improves compression factors by 7.3\% over FSST and 17.0\% over
\fsstxx{} on average, with decompression preserved for \optfsst{} and improved for
\optfsstxx{}.

The ablation study confirms that the three contributions are complementary: dynamic-programming-based encoding improves final compression decisions, triple counting accelerates the discovery of longer useful symbols, and pruning reduces wasted symbol-table capacity. 

\sparagraph{Future Work.} Future work should focus on closing the remaining compression-speed gap to \fsst{}. The most immediate direction is to further optimize \texttt{BuildDP}, especially the trie traversal and candidate lookup logic, since this is the dominant added cost of dynamic-programming-based encoding. A second direction is to explore more adaptive variants of the proposed techniques, for example enabling pruning or triple counting only when the sampled data suggests that they are beneficial.

\bibliographystyle{ACM-Reference-Format}
\bibliography{optfsst}

\appendix

\section{STS$_\ell$ NP-hardness}\label{appendix:np-hardness}

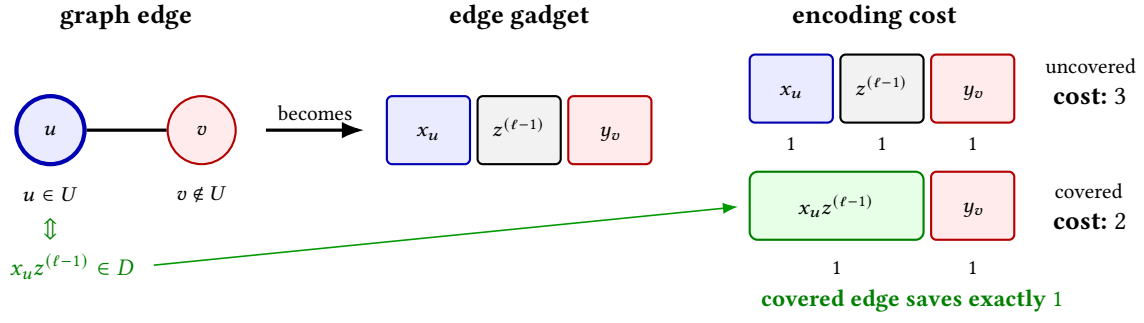
\begin{figure*}[t]
	\centering
	\begin{tikzpicture}[
		>=Latex,
		font=\small,
		vL/.style={circle,draw=blue!70!black,thick,fill=blue!8,minimum size=9mm},
		vR/.style={circle,draw=red!70!black,thick,fill=red!8,minimum size=9mm},
		byte/.style={
			rectangle,
			draw,
			thick,
			rounded corners=2pt,
			minimum width=11mm,
			minimum height=9mm,
			inner sep=1pt,
			text height=1.7ex,
			text depth=.5ex
		},
		byteL/.style={byte, fill=blue!8, draw=blue!70!black},
		byteX/.style={byte, fill=gray!10},
		byteR/.style={byte, fill=red!8,draw=red!70!black},
		phrase/.style={
			rectangle,
			draw=green!50!black,
			thick,
			rounded corners=3pt,
			fill=green!10,
			minimum width=23mm,
			minimum height=9mm,
			inner sep=1pt,
		},
		title/.style={font=\bfseries\large},
		lbl/.style={font=\small},
		greenlbl/.style={font=\bfseries\color{green!50!black}}
		]
		
		
		\node[title] at (0,2.7) {graph edge};
		
		\node[vL, ultra thick] (u) at (-1.0,1.2) {$u$};
		\node[vR]              (v) at ( 1.0,1.2) {$v$};
		
		\draw[very thick] (u) -- (v);
		
		\node[lbl] at (-1.0,0.35) {$u\in U$};
		\node[lbl] at ( 1.0,0.35) {$v\notin U$};
		
		\node[greenlbl] at (-1,-.1) {$\Updownarrow$};
		\node[greenlbl] (corr) at (-.75,-0.6) {$\; x_u z^{(\ell -1)}\in D$};
		
		
		\draw[->,very thick] (1.85,1.2) -- (3.15,1.2)
		node[midway,above] {becomes};
		
		
		\node[title] at (5.2,2.7) {edge gadget};
		
		\node[byteL] (ga) at (4.0,1.2) {$x_u$};
		\node[byteX] (gx) at (5.2,1.2) {$z^{(\ell -1)}$};
		\node[byteR] (gb) at (6.4,1.2) {$y_v$};
		
		
		
		\node[title] at (9.9,2.7) {encoding cost};
		
		\node[byteL] (ua) at (8.8,1.75) {$x_u$};
		\node[byteX] (ux) at (10,1.75) {$z^{(\ell -1)}$};
		\node[byteR] (ub) at (11.2,1.75) {$y_v$};
		
		\node[lbl] at (8.8,1.02) {$1$};
		\node[lbl] at (10,1.02) {$1$};
		\node[lbl] at (11.2,1.02) {$1$};
		
		\node[lbl, anchor=west] at (12.05,2.05) {uncovered};
		\node[title, anchor=west] at (12.15,1.65) {cost: $3$};
		
		\node[phrase] (pa) at (9.4,0.20) {$x_u z^{(\ell -1)}$};
		\node[byteR]  (pb) at (11.2,0.20) {$y_v$};
		
		\node[lbl] at (9.4,-0.62) {$1$};
		\node[lbl] at (11.2,-0.62) {$1$};
		
		\node[lbl, anchor=west] at (12.175,0.4) {covered};
		\node[title, anchor=west] at (12.15,0) {cost: $2$};
		
		\node[greenlbl] at (10.4,-1.05) {covered edge saves exactly $1$};
		
		
        \draw[
            ->,
            semithick,
            green!60!black,
            shorten >=4pt,
            shorten <=4pt
        ]
        (corr.east) -- (pa.west);
    \end{tikzpicture}
	
	\caption{
		In the reduction, each graph edge $(u,v)$ becomes the string gadget $x_u z^{(\ell -1)} y_v$.
		If $u\in U$, we add the length-$\ell$ string $x_u z^{(\ell -1)}$ to the dictionary $D$ (and similarly $v\in U \Leftrightarrow z^{(\ell -1)}y_v\in D$), which lowers the encoding cost of the edge gadget from $3$ to $2$.
	}
	\label{fig:reduction_stsell}
\end{figure*}

\MAINTHM*
\begin{proof}
	We reduce from the NP-hard decision version of Maximum Vertex Coverage on bipartite graphs (VCB) \cite{pvcb-2}.
	Given a bipartite graph $G=(L\cup R,E)$ and positive integers $K$ and $q$, decide whether there is a set $U\subseteq L\cup R$ with $|U|\leq K$ that covers at least $q$ edges, i.e., at least $q$ edges have at least one endpoint in $U$.
	
	Let $m\vcentcolon=|E|$ and let $\ell\geq 2$ be fixed.
	We construct an instance of $\mathrm{STS}_\ell$.
    
	Introduce one character $x_u$ for every $u\in L$, one character $y_v$ for every $v\in R$, and one extra character $z$.
	For every edge $(u,v)\in E$, we create an edge gadget $x_u z^{(\ell -1)} y_v$ of length $\ell+1$.
	Next we define
    \[
    C\vcentcolon=\{z^{(\ell -1)}\}\cup \{x_u: u\in L\}\cup \{y_v: v\in R\}.
    \]
	The multiset $\mathcal{T}$ consists of all edge gadgets, and additionally $m+1$ copies of each string $c\in C$ that we refer to as \emph{copy strings}.
    
	Finally, we set $K'\vcentcolon=|C|+K$ and $B'\vcentcolon=(m+1)|C|+3m-q$.
	We show that $(G,K,q)$ is a yes-instance of VCB if and only if $(\mathcal{T}, K',B')$ is a yes-instance of STS$_\ell$.
	
	\noindent\textbf{Forward direction.}
	Suppose there is a set $U\subseteq L\cup R$ with $|U|\leq K$ covering $r\geq q$ edges.
	Define
	\[
	D\vcentcolon=C\cup \{x_u z^{(\ell -1)}: u\in L\cap U\}\cup \{z^{(\ell -1)} y_v: v\in R\cap U\}
	\]
	with $|D|= |C|+|U|\leq |C|+K=K'$.
	Each of the $(m+1)|C|$ copy strings is itself contained in $D$ and can be encoded as a dictionary symbol at cost $1$.
	Hence the copy strings contribute a total cost of $(m+1)|C|$.
    
	Next, consider the edge gadget $x_u z^{(\ell -1)} y_v$.
	If $u\in U$, then it can be encoded as $x_u z^{(\ell -1)}$ and $y_v$ with cost $2$.
	Similarly, if $v\in U$, it can also be encoded with cost $2$.
	If neither $u$ nor $v$ are in $U$, the edge gadget can only be encoded as $x_u$, $z^{(\ell -1)}$, and $y_v$ with cost $3$.
	Therefore the $r$ covered edges contribute cost $2r$, and the remaining $m-r$ edges contribute cost $3(m-r)$.
    
	Thus,
	\begin{align*}
	\text{cost}_D(\mathcal{T})
    &=(m+1)|C|+2r+3(m-r)\\
    &=(m+1)|C|+3m-r\\
    &\leq (m+1)|C|+3m-q=B',
	\end{align*}
	so the constructed STS$_\ell$ instance is a yes-instance.
    See Figure~\ref{fig:reduction_stsell} for an illustration.
	
	\noindent\textbf{Reverse direction.}
	Suppose there is a dictionary $D$ with $|D|\leq K'$ and $\text{cost}_D(\mathcal{T})\leq B'$.
    
	First, we show that $C\subseteq D$.
	For contradiction, assume that some $c\in C$ is not in $D$.
	Then, each of the $(m+1)$ copies of $c$ in $\mathcal{T}$ must be encoded with a cost of at least $2$ instead of cost $1$.
	Thus, the copy strings contribute cost at least $(m+1)|C|+(m+1)$.
    
	Moreover, every edge gadget has length $\ell+1$, while dictionary symbols have length at most $\ell$.
	Therefore every edge gadget must be encoded with cost at least $2$.
	Hence the total cost is at least
	\begin{align*}
	\text{cost}_D(\mathcal{T})\geq&\,(m+1)|C|+(m+1)+2m\\
    =&\,(m+1)|C|+3m+1\\
    >&\,(m+1)|C|+3m-q=B',
	\end{align*}
	contradicting the assumption that $\text{cost}_D(\mathcal{T})\leq B'$.
	Thus $C\subseteq D$.
	
	Since $C\subseteq D$ and $|D|\leq |C|+K$, there are at most $K$ dictionary symbols not contained in $C$.
	As noted above, no edge gadget $x_u z^{(\ell -1)} y_v$ can be encoded with cost less than $2$.
	Moreover, the copy strings contribute cost at least $(m+1)|C|$.
    
	Hence
    \[
    \text{cost}_D(\mathcal{T})\leq B'=(m+1)|C|+3m-q
    \] implies that the total cost of edge gadgets is at most $3m-q$, and therefore at least $q$ edge gadgets must have cost exactly $2$.
	
	Every edge gadget $x_u z^{(\ell -1)} y_v$ with cost $2$ is split into two dictionary symbols $x_u z^{(t)}$ and $z^{(\ell-1-t)}y_v$ for some $t\in \{0,\dots,\ell-1\}$.
	At least one of these symbols is not in $C$, since $\ell\geq 2$.
	Define
	\begin{align*}
	U\vcentcolon=\,&\,
    \{u\in L: x_u z^{(t)}\in D\text{ for some }t\geq 1\}\\
    &\,\cup\, \{v\in R: z^{(t)}y_v\in D\text{ for some }t\geq 1\}.
	\end{align*}
	Each element in $D\setminus C$ contributes to at most one vertex of $U$, therefore $|U|\leq K$.
	Moreover, every edge gadget with cost $2$ has at least one endpoint in $U$.
	Since at least $q$ edge gadgets have cost $2$, the set $U$ covers at least $q$ edges.
    
	Hence the original VCB instance is a yes-instance.
	This concludes the proof.
\end{proof}

\section{STS$_\ell$ with fixed alphabet size }\label{appendix:STS_l_with_bounded_alphabet_size_in_P}
\begin{theorem}
    STS$_\ell$ can be solved in polynomial time for every fixed $\ell$ and fixed alphabet size $n\vcentcolon= |\Sigma|$.
\end{theorem}
\begin{proof}
Let $(\mathcal{T},K,B)$ be an instance over an alphabet $\Sigma$, where both $\ell$ and $n\vcentcolon= |\Sigma|$ are fixed.
Let $S\vcentcolon=\bigcup_{i=1}^{\ell}\Sigma^i$ be the set of all strings over $\Sigma$ of length at most $\ell$.
Since $|S|= \sum_{i=1}^\ell n^i\leq \ell n^\ell$ is constant, we can enumerate all $2^{|S|}$ subsets $D\subseteq S$ in constant time.
For each such subset $D$, we check whether $|D|\leq K$; if so, we compute  $\text{cost}_D(\mathcal{T})$ by dynamic programming in polynomial time.
The instance is a yes-instance if and only if some enumerated subset $D$ satisfies $|D|\leq K$ and
\[
\text{cost}_D(\mathcal{T})\leq B.
\]
Thus the problem can be solved in polynomial time.
\end{proof}

\end{document}